\documentclass[envcountsame,envcountsect]{llncs}
\usepackage{amsmath,amssymb,amsfonts,mathrsfs}
\usepackage[svgnames]{xcolor}
\usepackage{framed}
\usepackage{verbatim}
\usepackage[colorlinks=true,linkcolor=black,citecolor=MidnightBlue,urlcolor=MidnightBlue]{hyperref}
\usepackage{xspace}
\usepackage{fullpage}

\newcommand{\DEC}[0]{\operatorname{DEC}}
\newcommand{\charikar}[0]{{\sc CountSketch$_b$}} 
\newcommand{\countsketch}{\textsc{CountSketch}\xspace}
\newcommand{\freq}[0]{{\sc \bf ApproxFreqElements}}
\newcommand{\rarity}[0]{{\sc \bf ApproxRarity}}
\newcommand{\similarity}[0]{{\sc \bf ApproxSimilarity}}

\newcommand{\fk}[0]{\phi_k}

\begin{document}
\title{How to Catch $L_2$-Heavy-Hitters on Sliding Windows}

\author{
	Vladimir Braverman\inst{1}
 \and Ran Gelles\inst{2}
 \and Rafail Ostrovsky\inst{3}
}

\institute{
Department of Computer Science, Johns Hopkins University, \email{vova@cs.jhu.edu}.
\and
Department of Computer Science,
 University of California, Los Angeles,  \email{gelles@cs.ucla.edu}.
\and
Department of Computer Science and Mathematics,
   University of California, Los Angeles, \email{rafail@cs.ucla.edu}. \\
 }
\maketitle

\begin{abstract}
Finding heavy-elements (heavy-hitters) in streaming data is one of the central, and well-understood tasks.
Despite the importance of this problem, when considering the {\em sliding windows} model of streaming (where elements eventually expire) 
the problem of finding $L_2$-heavy elements has remained completely open 
despite multiple papers and considerable success in finding $L_1$-heavy elements. 

Since the $L_2$-heavy element problem doesn't satisfy certain conditions,
existing methods for sliding windows algorithms, such as 
smooth histograms or exponential histograms are not directly applicable 
to it.
In this paper, we develop the first polylogarithmic-memory algorithm for finding
$L_2$-heavy elements in the sliding window model.

Our technique 
allows us not only to find $L_2$-heavy
elements, but also heavy elements with respect to any $L_p$ with $0<
p \le 2$ on sliding windows. 
By this we completely ``close the gap'' and resolve the
question of finding $L_p$-heavy elements in the  sliding window model with
polylogarithmic memory, 
since it is well known that for $p> 2$ this task is impossible.  

We demonstrate a broader applicability
of our 
method on two additional examples: we show
how to obtain a sliding window approximation of 
the similarity of two streams, and of the fraction of elements
that appear exactly a specified number of times within the window
(the $\alpha$-rarity problem). In these two illustrative examples of
our method, we replace the current {\em expected} memory bounds with
{\em worst case} bounds. 

\end{abstract}

\thispagestyle{empty}

\section{Introduction} \label{sec:intro}

A {\em data stream} $S$ is an ordered multiset of elements $\{a_0, a_1, a_2 \ldots \}$
where each element $a_t \in \{1, \ldots, u\}$ arrives at time $t$.
In the {\em sliding window model} we consider
at each time $t \ge N$ the last $N$ elements of the stream, i.e.\ the window
$W=\{a_{t-(N-1)}, \ldots,    a_{t} \}$.
These elements are called  {\em active}, whereas
elements that arrived prior to the current window
$\{ a_i \ \vert\  0\le i < t-(N-1)\}$ are {\em expired}.
For $t < N$, the window consists of all the elements received so far, $\{a_0, \ldots,    a_{t} \}$.

Usually, both $u$ and $N$ are considered to be extremely large so it is not applicable to save the entire stream (or even one entire window) in memory. The problem is to be able to calculate various characteristics about the window's elements using 
small amount of memory (usually, polylogarithmic in $N$ and $u$).
We refer the reader to the books of Muthukrishnan~\cite{muth05} and
Aggarwal (ed.)~\cite{aggarwal07}
for extensive surveys on data stream models and algorithms.

One of the main open problems in data streams deals with the relations between
the different streaming models~\cite{IITK06}, specifically between the unbounded stream model
and the sliding window model. In this paper we provide  another important step in clarifying
the connection between these two models by showing that finding $L_p$-heavy hitters
is just as doable on sliding windows as on the entire stream.

We focus on approximation-algorithms for
certain statistical characteristics of the data streams, specifically,
finding frequent elements.
The problem of finding frequent elements in a stream is useful
for many applications, such as network monitoring~\cite{SW02}
and DoS prevention~\cite{EV03,CKMS03,BAA07},
and was extensively explored over the last decade
(see~\cite{muth05,CH08} for a definition of the problem and a survey of
existing solutions, as well as~\cite{CCF02,MM02,GDDLM03,JQSYZ03,CM04,AM04,CM05hot,ZG08,HT08}).
%

We say that an element is \emph{heavy}
if it appears more times than a constant fraction of some $L_p$ norm of the stream.
Recall that for $p>0$, the $L_p$ norm of the frequency vector\footnote{%
Throughout the paper we use the term ``$L_p$ norm'' to indicate the $L_p$ norm of the frequency
vector, i.e., the $p$th root of the $p$th frequency moment $F_p=\sum_i n_i^p$~\cite{AMS99}, rather than the norm of the data itself.}
is defined by $L_p = (\sum_i n_i^p)^{1/p}$, where $n_i$ is
the frequency of element $i\in[u]$, i.e., the number of times $i$ appears
in the window. 
Since different $L_p$ can be considered, 
we obtain several different ways to define a ``heavy'' element.
Generally speaking (as mentioned in~\cite{Indyk09}), 
when considering frequent elements (heavy-hitters) with respect to~$L_p$,
the higher $p$ is, the better. Specifically, 
identifying frequent elements with respect to $L_2$ is better than~$L_1$
since an $L_1$~algorithm can always be replaced with an $L_2$~algorithm, with less or equal memory consumption (but not vice versa).

Naturally, finding frequent elements with respect to the $L_2$ norm is
a more difficult task (memory-wise) than the equivalent $L_1$ problem.
To demonstrate this fact let us regard the following example:
let $S$ be a stream of size $N$, in which
the element $a_1$ appears $\sqrt{N}$ times,
while the rest of the elements $a_2, \ldots, a_{N-\sqrt{N}}$ appear
exactly once in $S$. Say we wish to identify $a_1$ as an heavy element.
Note that  $n_1 = \frac{1}{\sqrt{N}}L_1$ while $n_1= cL_2$, 
where $c$ is a constant, lower bounded by  $\frac{1}{\sqrt{2}}$.
Therefore, as $N$ grows, $n_1/L_1\to 0$ goes to zero, while $n_1/L_2$ is bounded by a constant.
If an algorithm finds elements which are heavier than $\gamma L_p$
with memory $poly(\gamma^{-1},\log N,\log u)$, then for $p=2$ we get  a 
polylogarithmic memory, while for $p=1$ the memory consumption is super-logarithmic.

We focus on solving the following $L_2$-heaviness problem:
\begin{definition}[$(\gamma,\epsilon)$-approximation of $L_2$-frequent elements]\label{def:L2heavy}  
For $0 < \epsilon , \gamma < 1$,
output any element $i\in[u]$
such that $n_i > \gamma L_2$ and no element such that $n_i < (1-\epsilon)\gamma L_2$.
\end{definition}
\noindent The $L_2$ norm is the most powerful norm for which we can expect
a polylogarithmic solution, for the frequent-elements problem. This is
due to the known lower bound of $\Omega(u^{1-2/p})$ for calculating $L_p$
over a stream~\cite{SS02,BJKS02}.
%

There has been a lot of progress on the question of
finding  $L_1$-frequent elements, in the sliding window model~\cite{AM04,ZG08,HT08},
however those algorithms cannot be used to find $L_2$-frequent elements with an efficient memory.
In 2002, Charikar, Chen and Farach-Colton~\cite{CCF02} developed the \countsketch algorithm
that can approximate the ``top $k$'' frequent-elements on \emph{an unbounded stream},
where $k$ is given as an input. Formally, their algorithm outputs
only elements with frequency larger than $(1-\epsilon)\fk$, where $\fk$ is the frequency
of the $k$th most frequent element in the stream, using memory proportional
to $L_2^2 / (\epsilon\fk)^2$.
Since the ``heaviness'' in this case is relative to  $\fk$, and the memory
is bounded by the fraction $L_2^2 / (\epsilon\fk)^2$,
Charikar et al.'s algorithm finds in fact heaviness in terms
of the $L_2$ norm.
A natural question is whether one can develop an algorithm for
finding frequent-elements that appear at least $\gamma L_2$ times in the
{\em sliding window model}, using $poly(\gamma^{-1},\log N,\log u)$ memory.

\medskip\noindent\textbf{Our Results.}
We give the first polylogarithmic algorithm 
for 
finding an $\epsilon$-approximation of the $L_2$-frequent elements
in the sliding window model.
Our algorithm is able
to identify elements that appear within the window a number of times which is
at least a $\gamma$-fraction of the $L_2$ norm of the window, up to a
multiplicative factor of $(1-\epsilon)$. In addition,
the algorithm guarantees to output {\em all\/} the elements
with frequency at least $(1+\epsilon)\gamma L_2$.
\begin{theorem}\label{thm:main}
There exists an efficient sliding window algorithm that outputs a $(\gamma,\epsilon)$-approximation of the $L_2$-frequent-elements, with probability at least~$1-\delta$ and memory~$poly(\epsilon^{-1},\gamma^{-1}, \log N, \log \delta^{-1})$.
\end{theorem}

We note that the \countsketch algorithm 
works in the unbounded model and does not apply directly on sliding windows.
Moreover, \countsketch  solves a slightly different (yet related) problem,
namely, the top-$k$ problem, rather than the $L_2$ heaviness.
To achieve our result on $L_2$ heavy hitters, 
we combine in a non-trivial way
the scheme of  Charikar et al.\@ with a sliding-window approximation
 for $L_2$ as given by Braverman and Ostrovsky~\cite{BO07}.
Variants of these techniques sufficient to derive similar results were known since $2002$,\footnote{Indeed, we use the algorithm of Charikar et al.~\cite{CCF02} that is known since 2002. Also, it is possible to replace (with some non-trivial effort) our smooth histogram method for $L_2$ computation with the algorithm of Datar, Gionis, Indyk and Motwani~\cite{DGIM02} for $L_2$ approximation.}
however no algorithm for $L_2$ heavy hitters was reported despite several papers on $L_1$ heavy hitters.

Our solution gives another step in the direction of making
a connection between the unbounded and sliding window models,
as it provides an answer for the very important question
of heavy hitters in the sliding window model. The result joins
the various solutions of finding $L_1$-heavy hitters in sliding windows~\cite{GDDLM03,AM04,NL05,BAA07,ZG08,HT08,HT10},
and can be used in various algorithms that require identifying $L_2$ heavy hitters, such as~\cite{IW05,BGKS06} and others. 
More generally, our paper resolves the question of finding $L_p$-heavy elements on sliding windows for all values of~$p$ that allows small memory one-pass solutions (i.e.\ for $0< p\leq 2$). By this we completely close the gap between the case of~$p\le 1$, solved by previous works,  and the impossibility result for the case of~$p>2$.

\smallskip\noindent\textbf{A Broader Perspective. }
In fact, one can consider the tools we develop for the frequent elements problem
as a general method that allows obtaining a sliding window solution
out of an algorithm for the unbounded model, for a wide
range of functions. We explain this concept in this section.

Many statistical properties were aggregated into families, and efficient
algorithms were designed for those families. For instance,
Datar, Gionis, Indyk and Motwani, in their seminal paper~\cite{DGIM02} showed that
a sliding window estimation is easy to
achieve for any function which is {\em weakly-additive}
by using a data structure named {\em exponential histograms}~\cite{DGIM02};
for certain functions that decay with time, one can maintain time-decaying aggregates~\cite{CS06};
another data structure, named {\em smooth-histogram}~\cite{BO07} can be used in order
to approximate an even larger set of functions, known as {\em smooth functions}.
See~\cite{aggarwal07} for a survey of synopsis construction.

In this paper we introduce a new concept which uses a smooth-histogram
in order to perform sliding window approximation of {\em non-smooth} properties.
Informally speaking, the main idea is to relate the
non-smooth property $f$ with some other, smooth\footnote{Of course, other kinds of aggregations can be used, however our focus is on smooth histograms.}, property $g$,
such that changes in $f$ are bounded by the changes in $g$.
By maintaining a smooth-histogram for the smooth function $g$, we {\em partition} the
stream into sets of sub-streams (buckets).
Due to the properties of the smooth-histogram we can bound the
error (of approximating~$g$) for every sub-stream, and thus get an approximation of~$f$.
We use the term {\em semi-smooth} to describe these kinds of algorithms.

We demonstrate the above idea by showing a concrete efficient sliding window
algorithm for the properties
of rarity and similarity~\cite{DM02}; we stress that neither is smooth
(see Section~\ref{sec:simrar} for definitions of these problems).
Although there already exist algorithms for these problems with   {\em expected} polylogarithmic memory~\cite{DM02},
our techniques improve these results and obtain a {\em worst case} memory consumption of essentially the same magnitude (up to a factor of~$\log\log N$).

In addition to the properties of rarity and similarity, we believe that the tools we develop
here can be used to build efficient sliding window approximations for many other (non-smooth)
properties and provide a general new method for computing on sliding windows.
Indeed, in a subsequent work Tirthapura and Woodruff~\cite{TW12} use our methods to compute various \emph{correlated aggregations}.
It is important to note that trying to build a smooth-histogram (or any other known sketch)
directly to $f$ will not preserve the
required invariants,
and the memory consumption might
not be efficient.

\bigskip
\noindent\textbf{Previous Works.}
\vspace{-3ex}
\paragraph{Frequent elements.}
Finding elements that appear many times in the stream (``heavy hitters'')
is a very central question and thus has been extensively studied
both for the unbounded model~\cite{DLM02,KSP03,CM04,MAA05}
and for the sliding window model~\cite{AM04,NL05,ZG08,HT08} as well as other variants
such as the {\em offline stream} model~\cite{MM02},
insertion and deletion model~\cite{CM05hot,JQSYZ03}, finding
heavy-distinct-hitter~\cite{BAA07}, etc.
Reducing the processing time was done by~\cite{LT06} into $O(\frac1\epsilon)$
and by~\cite{HT10} into $O(1)$.

Another problem which is related to finding the heavy hitters,
is the top-$k$ problem, namely, finding the $k$ most frequent elements.
As mentioned above, Charikar, Chen and Farach-Colton~\cite{CCF02} provide an algorithm that
finds the $k$ most frequent elements in the unbounded model (up to a precision of $1\pm\epsilon$).
Golab, DeHaan, Demaine, L\'{o}pez-Ortiz and Munro~\cite{GDDLM03}
solve this problem in the {\em jumping window} model.

\paragraph{Similarity and $\alpha$-rarity.}
The similarity problem was defined in order to give a rough estimation
of closeness between files over the web~\cite{BGMZ97} (and independently in~\cite{Cohen97}).
Later, it was shown how to use min-hash functions~\cite{Indyk99} in order to sample from the stream, and estimate the similarity of two streams.

The notion of $\alpha$-rarity, introduced by Datar and Muthukrishnan~\cite{DM02},
is that of finding
the fraction of elements that appear exactly $\alpha$ times within the stream.
This quantity can be seen as finding
the fraction of elements with frequency within certain bounds.

The questions of rarity and similarity were analyzed, both for
the unbounded stream and the sliding window models, by Datar and
Muthukrishnan~\cite{DM02}, achieving an expected memory bound of
$O(\log N + \log u)$ words of space for a constant $\epsilon, \alpha, \delta$.
At the bit level, their algorithm requires
$O(\alpha\cdot\epsilon^{-3}\log\delta^{-1}\log N(\log N+\log u))$ bits for $\alpha$-rarity and
$O(\epsilon^{-3}\log\delta^{-1}\log N(\log N+\log u))$ bits
for similarity, with $1-\delta$ being the probability of success\footnote{These bounds
are not explicitly stated in~\cite{DM02}, but follow from the analysis (see Lemma~1
and Lemma~2 in~\cite{DM02}).}.

\section{Preliminaries}\label{sec:pre}

\subsection{Notations}
We say that an algorithm $A_f$ is an $(\epsilon,\delta)$-approximation of a function $f$, if for any input $S$, 
$(1-\epsilon)f(S) \le A_f(S) \le (1+\epsilon)f(S)$, except with probability $\delta$ over $A_f$'s coin tosses.
We denote this relation as $A_f\in (1\pm\epsilon)f$ for short.
We denote an output of an approximation
algorithm with a hat symbol, e.g., the estimator of $f$ is denoted $\hat f$.

The set $\{1,2, \dotsc, n\}$ is usually denoted as $[n]$.
If a stream $B$ is a suffix of $A$, we denote $B \subseteq_r A$. For instance,
let $A=\{q_1, q_2, \ldots, q_n\}$ then $B=\{q_{n_1}, q_{n_1+1}, \ldots, q_n \} \subseteq_r A$,
for $1 \le n_1 \le n$.
The notation $A\cup C$ denotes the concatenation
of the stream $C=\{c_1, c_2, \dotsc, c_m\}$ to the end of stream
$A$, i.e., $A\cup C= \{q_1, q_2, \ldots, q_n, c_1, c_2, \ldots c_m\}$.
The notation $|A|$ denotes the number of different elements in the stream $A$, that is
the cardinality of the {\em set} induced by the multiset $A$.
The size of the stream (i.e.\ of the multiset) $A$
will be denoted as $\|A\|$, e.g., for the example above $\|A\|=n$.

We use the notation $\tilde O(\cdot)$ to indicate an asymptotic bound
which suppresses terms
of magnitude\\ $poly(\log\frac1\epsilon,\log\log\frac1\delta,\log\log N, \log\log u)$.

\subsection{Smooth histograms}
Recently, Braverman and Ostrovsky~\cite{BO07} showed that a function
$f$ can be $\epsilon$-approximated in the sliding window model, if 
$f$ is a \emph{smooth function}, and if it can be calculated (or approximated) 
in the unbounded stream model.
Formally,
\begin{definition}\label{smooth} 
A polynomial function $f$ is $(\alpha,\beta)$-smooth if it satisfies
 the following properties:
(i) $f(A)\ge0$; (ii) $f(A) \ge f(B)$ for $B\subseteq_r A$; and
(iii)  there exist $0<\beta\le\alpha<1$ such that if $(1-\beta)f(A) \le f(B)$ for $B\subseteq_r A$, then
$(1-\alpha)f(A\cup C)\le f(B\cup C)$ for any~$C$.
\end{definition}
\noindent
If an $(\alpha,\beta)$-smooth $f$ can be calculated
(or ($\epsilon,\delta$)-approximated) on an unbounded stream
with memory $g(\epsilon,\delta)$, then there exists an
($\alpha+\epsilon,\delta$)-estimation of $f$ in the sliding window model using
$O(\frac{1}{\beta}\log N(g(\epsilon,\frac{\delta\beta}{\log N})+\log N))$ bits~\cite{BO07}.

The key idea is to construct a ``smooth-histogram'', 
a structure that contains estimations on $O(\frac{1}{\beta}\log N)$-suffixes of the stream, $A_1 \supseteq_r A_2 \supseteq_r \ldots \supseteq_r A_{c\frac{1}{\beta}\log(n)}$.
Each suffix $A_i$ is called {\em a Bucket}.  Each new element in the stream initiates a new bucket, however adjacent buckets with a close estimation value are removed (keeping only one representative).
Since the function is ``smooth'', i.e., monotonic and slowly-changing, it is enough to save $O(\frac{1}{\beta}\log N)$ buckets in order to maintain a reasonable approximation of the window. At any given time, the current window $W$ is between buckets $A_1$ and $A_2$, i.e. $A_1 \supseteq_r W \supseteq_r A_2$. Once the window ``slides'' and the first element of $A_2$ expires, we delete the bucket $A_1$ and renumber the indices so that $A_2$ becomes the new $A_1$, $A_3$ becomes the new $A_2$, etc. 
We use the estimated value of bucket $A_1$ to estimate the value of the current window. The relation between the value of~$f$ on the window and on the first bucket
is given by
$
(1-\alpha)f(A_1)\le f(A_2) \le f(W) \le f(A_1)\ .
$

\section{A Semi-Smooth Estimation of Frequent Elements}\label{sec:freq}

In this section we develop an efficient semi-smooth algorithm for finding elements
that occur frequently within the window.
Let $n_i$ be {\em the frequency}
of  element $i\in \{1, \dotsc, u\}$, i.e.,  the number of times $i$ appears in the window.
The {\em first frequency norm} and the {\em second frequency norm} of the window are defined by
$L_1= \sum_{i=1}^u n_i=N$ and $L_2 = \left( \sum_{i=1}^u n_i^2 \right)^\frac12$.
In many previous works,  (e.g.,~\cite{CM04,AM04,muth05,ZG08,HT08})
the task of finding heavy-elements is defined using the $L_1$ norm as follows,
\begin{definition}[$(\gamma, \epsilon)$-approximation of $L_1$-heavy hitters]\label{def:L1heavy}
Output any element $i\in[u]$
such that ${n_i \ge \gamma L_1}$ and no element such that ${n_i \le (1-\epsilon)\gamma L_1}$.
\end{definition}

Our notion of approximating {\em frequent elements} is given by Definition~\ref{def:L2heavy}.
An equivalent definition 
which we use in our proof is the following:
\begin{definition}\label{def:L2heavyVar}
For $0 < \epsilon , \gamma < 1$,
output all elements $i\in[u]$
with frequency
higher than $(1+\epsilon)\gamma L_2$, and do not output any element with
frequency lower than $(1-\epsilon)\gamma L_2$. 
\end{definition}
Observe that the~$L_2$ approximation is stronger than the above $L_1$
definition. If an element 
is heavy in terms of $L_1$ norm,
it is also heavy in terms of the $L_2$ norm,
\[
n_i \ge \gamma L_1 = \gamma \sum_j n_j \quad \Longrightarrow \quad n_i^2 \ge \gamma^2 \Big(\sum_j n_j\Big)^2 \ge \gamma^2 \sum_j n_j^2 = (\gamma L_2)^2\ ,
\vspace{-3pt}
\]
while the opposite direction does not apply in general.

In order to identify the frequent elements in the current window, use a variant of  
the \countsketch algorithm of Charikar et~al.~\cite{CCF02}, which provides an
$\epsilon$-approximation (in the unbounded stream model) for the following 
top-frequent approximation problem.
\begin{definition}[$(k,\epsilon)$-top frequent approximation]\label{def:ktop}
Output a list of $k$~elements
such that every element~$i$ in the output 
has a frequency $n_i > (1-\epsilon)\fk$, where
$\fk$ is the frequency of the $k$-th most frequent element in the stream.
\end{definition}

The \countsketch algorithm 
guarantees that any element
that satisfies $n_i > (1+\epsilon)\fk$, appears in the output.
This algorithm runs on a stream of size $n$ and
succeeds with probability at least $1-\delta$, and memory complexity
of
\(
O\left ( \left(k+\frac1{(\epsilon\gamma)^2}\right)\log\frac{n}\delta\right )%
\),
 for every $\delta >0$, given that $\fk \ge \gamma L_2$.

Definition~\ref{def:ktop} and Definition~\ref{def:L2heavy} do not describe the same problem, yet
they are strongly connected. In fact, our method allows solving the frequent elements problem
under both definitions, however
in this paper we focus on solving the $L_2$-frequent-elements problem,
as defined by Definition~\ref{def:L2heavyVar}.
In order to do so,
we use a variant of the \countsketch algorithm 
with specific parameters tailored for our problem
(See full details in Appendix~\ref{sec:varCharikar}).
This variant outputs a list of elements, and is guaranteed to output
every element with frequency at least $(1+\epsilon')\gamma L_2$
and no element of frequency less than $(1-\epsilon')\gamma L_2$,
for an input parameter~$\epsilon'$.

We stress that  \countsketch  
is not sufficient on its own
to prove Theorem~\ref{thm:main}. 
The main reason is that this algorithm works in the unbounded
stream model, rather than in the sliding window model. 
Another reason is that it must be tweaked in order not to output false positives. 
Our solution below makes a use of smooth-histograms to overcome these issues.

\subsection{Semi-smooth algorithm for frequent elements approximation}\label{sec:freqalg}
We construct
a smooth-histogram for the $L_2$ norm, and partition the stream into buckets accordingly.
It is known that the $L_2$ property is a
$(\epsilon,\frac{\epsilon^2}{2})$-smooth function~\cite{BO07}.
Using the method of Charikar et al.~\cite{CCF02}, separately on each bucket,
with a careful choice of parameters,
we are able to approximate
the $(\gamma,\epsilon)$-frequent elements problem on a sliding window (Fig.~\ref{alg:freq}).

\begin{figure}[!ht]
\normalsize
\begin{framed}
{\bf \freq$(\gamma,\epsilon,\delta)$}
\small
\begin{enumerate}
\item
Maintain an ($\frac\epsilon2,\frac\delta2$)-estimation of the $L_2$ norm of the window,
using a smooth-histogram.
\item \label{step:charikar}
For each bucket of the smooth-histogram, $A_1, A_2, \ldots$ maintain an
approximated list of the $k=\frac1{\gamma^2}+1$ most frequent elements,
by running $(\gamma, \frac\epsilon4,\frac\delta2)-${\charikar}.\\
(see \charikar's description in Appendix~\ref{sec:varCharikar}).
\item Let $\hat L_2$ be the approximated value of the $L_2$ norm of the current window $W$,
as given by the the smooth-histogram.
Let $q_1, \dotsc, q_k \in \{1, \dotsc, u\}$ be the list of the $k$ most heavy elements
in $A_1$, along
with $\hat n_1, \dotsc, \hat n_k$ their  estimated frequencies, as outputted by \charikar.
\item\label{step:filter} Output any element $q_i$ that satisfies
$\hat n_i > \frac1{1+\epsilon}\gamma \hat L_2$.
\end{enumerate}
\end{framed}
\vspace{-10pt}
\caption{A semi-smooth algorithm for the frequent elements problem}\label{alg:freq}
\vspace{-10pt}
\end{figure}

\begin{theorem}\label{thm:freq}
The semi-smooth algorithm \freq\ (Fig.~\ref{alg:freq}) is a
$(\gamma,O(\epsilon))$-approximation of the $L_2$-frequent elements problem,
with success probability at least $1-\delta$.
\end{theorem}

\begin{proof}
Recall that the smooth-histogram data structure of the $L_2$ guarantees us
an estimation $\hat L_2$ which is $(1\pm\epsilon)L_2(W)$; in addition there
exists some $\alpha$ such that $(1-\alpha)L_2(A_1) \le L_2(W) \le L_2(A_1)$.
In our
case the inequality is satisfied for $\alpha = \epsilon/2$
(see Theorem~3 and Definition~3 in~\cite{BO07}).
Any element $j$ with frequency $n_j(W) > (1+\epsilon)\gamma L_2(W)$
satisfies
\[
n_j(A_1) \ge n_j(W) \ge (1+\epsilon)\gamma L_2(W) \ge (1+\epsilon)(1-\epsilon/2)\gamma L_2(A_1)\ ,
\]
and will be added to the output list in Step~\ref{step:charikar}, since
Proposition~\ref{prop:allHigh} guarantees that
any element $i$ such that $n_i(A_i) > (1+\epsilon/4)\gamma L_2(A_1)$ is identified
by \charikar\ (assuming $\epsilon < \frac12$).

In order to show that all of the required elements survive  Step~\ref{step:filter}, we
use Lemma~\ref{lem:varCharikarFreqApprox} to bound the estimated frequency
$\hat n_i$ reported by \charikar, and show it is above the required threshold.
If $n_i(W) > (1+\epsilon)\gamma L_2(W)$ then
\[
\hat n_i(A) > n_i(A) - \frac\epsilon8\gamma L_2(A) > n_i(W) - \frac\epsilon{2-\epsilon}L_2(W) >
\left[1+\epsilon-\frac{\epsilon}{2-\epsilon} \right]\gamma L_2(W)\ ,
\]
recalling that
$\hat L_2< (1+\epsilon)L_2(W)$ implies that
the element survives Step~\ref{step:filter}.

While we are guaranteed that all the
$(1+\epsilon)\gamma L_2(W)$-frequent elements appear in the output list, it might
contain other elements which are not heavy enough. We now prove
that Step~\ref{step:filter} eliminates any element of frequency less than $(1-c\epsilon)\gamma L_2(W)$, for a constant $c$.

\begin{lemma}\label{lem:windowBucket}
If for an element $i$ there exists some $\zeta > \sqrt{\epsilon}$ such that $n_i(A_1) > \zeta L_2(A_1)$, then there exist a constant $\xi >0$ such that
$n_i(W) > \xi L_2(W)$.
\end{lemma}
\begin{proof} By the properties of the smooth-histogram,
\begin{align*}
L_2(W)^2 &> (1-\epsilon/2)^2L_2(A_1)^2 > (1-\epsilon) L_2(A_1)^2\\
n_i(W)^2 + \sum_{j\ne i}n_j(W)^2 &>
      n_i(A_1)^2 + \sum_{j\ne i}n_j(A_1)^2  -\epsilon L_2(A_1)^2\\
n_i(W)^2 &>  n_i(A_1)^2 - \epsilon L_2(A_1)^2 > (\zeta^2-\epsilon)L_2(A_1)^2
\end{align*}
and $n_i(W) > \xi L_2(W)$ for $\xi \le \sqrt{(\zeta^2-\epsilon)}$.
\qed
\end{proof}

Suppose some element~$i$ survives Step~\ref{step:filter}, then
$\hat n_i(A_1) > \frac1{1+\epsilon}\gamma \hat L_2 > \frac{1-\epsilon}{1+\epsilon}\gamma L_2(W)$. By Lemma~\ref{lem:varCharikarFreqApprox},
\[
n_i(A_1) \ge \hat n_i(A_1) - \frac\epsilon8\gamma L_2(A_1)
\ge \left (\frac{(1-\epsilon)(1-\frac\epsilon2)}{1+\epsilon}-\frac\epsilon8 \right)\gamma L_2(A_1) > (1-3\epsilon)\gamma L_2(A_1),
\]
and by Lemma~\ref{lem:windowBucket},
$n_i(W) \ge  \sqrt{1-7\epsilon}\cdot\gamma L_2(W)$.
This proves
that for small enough $\epsilon$ there exists some constant~$c$ such that
 the algorithm doesn't output any element with frequency lower than
$(1-c\epsilon)\gamma L_2(W)$.

To conclude, except for probability $\delta/2$ we are able to partition the stream into
$L_2$-smooth buckets, and except for probability $\delta/2$, the \charikar\
algorithm outputs a list which
can be used to identify the frequent elements of the window.
Using a union bound we conclude that the entire
algorithm succeeds except with probability $\delta$. This completes the proof of the theorem.
\qed
\end{proof}

\smallskip\noindent\textbf{Memory Usage.}
The memory usage of the protocol is composed of two parts: maintaining
a $(\epsilon/2, \delta/2)$-smooth-histogram of $L_2$, and running \charikar\
on each of the buckets.
According to \cite{BO07} (corollary~5),
maintaining a smooth-histogram for $L_2$ can be done with memory
$$O\left(\tfrac1{\epsilon^2}\log^2 N +
\tfrac1{\epsilon^4}\log N\log \frac{\log N}{\delta\epsilon}\right)$$
for a relative error of $\epsilon/2+\epsilon^2/8$,
with success probability at least $1-\delta/2$.
For  small enough $\epsilon$ we have
 $\epsilon/2+\epsilon^2/8 < \epsilon$ as required.

As for the second part, recall that one instance of \charikar\  requires  a memory of
$O\big(\frac{1}{\epsilon^2\gamma^2}\log \frac{n}{\delta}\big)$ (see Appendix~\ref{sec:varCharikar}),
where $n$ is the size of the input. In our case the maximal size of the input is
the size of the first bucket, $\|A_1\|$. Note that $\log \|A_1\| = O(\log N)$
since $(1-\alpha)L_2(A_1) \le L_2(W) \le N$.
The number of \charikar\ instances is bounded by the number of buckets,
$O(\frac1{\epsilon^2}\log N)$~\cite{BO07}, which leads to a total memory bound of
\[
O \left (  \frac{1}{\gamma^2\epsilon^4}\log N\log \frac{N}{\delta} +
             \frac1{\epsilon^4}\log N\log\frac1{\epsilon}\right )\ .
\]

\subsection{Extensions to any~$L_p$ with~$p< 2$}
It is easy to see that the same method can be used in order to approximate $L_p$-heavy
elements for any $0<p<2$, up to a $1\pm\epsilon$ precision.
The algorithms and analysis remain the same, except for using a smooth-histogram
for the $L_p$ norm, and changing the parameters by constants.
\begin{theorem}
For any $p\in (0,2]$, there exists a sliding window algorithm that outputs all the elements with frequency
at least $(1+\epsilon)\gamma L_p$, and no element with frequency less then
$(1-\epsilon)\gamma L_p$.
The algorithm succeeds with probability at least $1-\delta$ 
 and takes $poly(\epsilon^{-1}, \gamma^{-1}, \log N,\log \delta^{-1})$ memory.
\end{theorem}

\section{Estimation of Non-Smooth Properties Relativized to the Number of Distinct Elements}
\label{sec:simrar}
In this section we extend the method shown above and apply it to other non-smooth functions.
In contrast to the smooth $L_2$ used above,
in this section we use a different smooth function to partition the stream, namely
the distinct elements count problem. This allows us to obtain efficient
semi-smooth approximations for the (non-smooth) \emph{similarity} and \emph{$\alpha$-rarity} tasks.

\subsection{Preliminaries} 
We now show that counting the number of distinct elements in a stream is smooth.
This allows us to partition the stream into a smooth-histogram structure,
where each two adjacent buckets have
approximately the same number of distinct elements.
\begin{proposition}
Define $\DEC(A)$ as the number of distinct elements in the stream $A$, i.e., $\DEC(A)=|A|$.
The function $\DEC$ is an $(\epsilon,\epsilon)$-smooth-function, for every $0 \le \epsilon \le 1$.
\end{proposition}
\begin{proof}
Properties (i) and (ii) of Definition~\ref{smooth} follow directly from $\DEC$'s definition. As for property (iii),
assume that $B\subseteq_r A$ and ${(1-\epsilon)}\DEC(A) \le \DEC(B)$, then
\begin{eqnarray*}
{(1-\epsilon)}\DEC(A\cup C)&=&(1-\epsilon)\left[\DEC(A)+\DEC(C\setminus A) \right ]\\
&\le & \DEC(B) + (1-\epsilon)\DEC(C\setminus A) \\
&\le & \DEC(B) + \DEC(C\setminus B) \\   
&=& \DEC(B\cup C),
\end{eqnarray*}
where ``$A\setminus B$" represents the set of all the elements in $A$ which are not in $B$.
\qed
\end{proof}

There have been many works on counting distinct elements in streams,
initiated by Flajolet and Martin~\cite{FM83}, and later improved
by many others~\cite{AMS99,GT01,BKS02,BJKST02}.
Recently, Kane, Nelson and Woodruff provided an optimal algorithm
for $(\epsilon,\delta$)-approximating the number of distinct elements~\cite{KNW10},
using 
$O((\frac1{\epsilon^2}+\log u)\log\frac1\delta)$ bits and $O(1)$ time.
We use the method of Kane et al.\ in order to construct a smooth-histogram
for the distinct elements count with memory
$\tilde O\big((\log u+\frac{1}{\epsilon^2})\frac{1}{\epsilon}
\log N\log\frac1\delta+ \frac{1}{\epsilon}\log^2 N\big)$, suppressing
$\log\log N$ and $\log\frac1\epsilon$ terms.

Another tool we use is  {\em min-wise hash functions}~\cite{Broder97,BCFM00}, used in various algorithms in order to estimate different characteristics of data streams, especially the {\em similarity} of two streams~\cite{Broder97}.
Informally speaking,
these functions have a meaning of uniformly sampling an element from the stream, which makes them
a very useful tool.
\begin{definition}[min-hash]\label{min-hash} 
Let $\Pi = \{ \pi_i \}$ be a  family of permutations over $[u]=\{1,\ldots, u\}$. For a subset $A\subseteq [u]$ define $h_i$ to be the minimal permuted value of $\pi_i$ over $A$,
\(
h_i = \min_{a \in A} \pi_i(a).
\)
A family $\{ h_i \}$ of such functions is called exact min-wise independent hash functions (or min-hash)  if for any subset $A\subseteq [u]$ and $a\in A$,
\[ 
\Pr_i\left[h_i(A) = \pi_i(a)\right] = \frac{1}{|A|}. 
\]
The family $\{ h_i \}$ is called $\epsilon$-approximated min-wise independent hash functions (or $\epsilon$-min-hash) if for any subset $A\subseteq [u]$ and $a\in A$,
\[
\Pr_i\left[h_i(A) = \pi_i(a)\right] \in \frac{1}{|A|}(1\pm\epsilon).
\]
\end{definition}
A specific construction of $\epsilon$-min-hash functions was presented
by Indyk~\cite{Indyk99}, using only $O(\log\frac{1}{\epsilon}\log u)$ bits.
The time per hash calculation is bounded by $O(\log \frac{1}{\epsilon})$.
Min-hash functions can be used in order to estimate the similarity
of two sets, by using the following lemma,
\begin{lemma}{\rm (\cite{BCFM00}.~See also \cite{DM02}.)}\label{probAW}
For any two sets $A$ and $W$ and an $\epsilon'$-min-hash function~$h_i$, it holds that
\(
\Pr_i\left[h_i(A)=h_i(W)\right] = \frac{|A\cap W|}{|A\cup W|} \pm \epsilon'.
\)
\end{lemma}

\subsection{A semi-smooth estimation of $\alpha$-rarity}\label{sec:rarity}

In the following section we present an algorithm that estimates the $\alpha$-{\em rarity}
of a stream (in the sliding window model), i.e., the ratio of elements
that appear exactly $\alpha$ times in the window.
The rarity property  is known not to be smooth, yet
by using a smooth-histogram for distinct elements count, we are able to
partition the stream into $O(\frac1\epsilon \log N)$ buckets, and estimate the $\alpha$-rarity
in each bucket.

\begin{definition}
An element $x$ is {\em $\alpha$-rare} if it appears exactly $\alpha$ times in the stream.
The $\alpha$-rarity measure, $\rho_\alpha$, denotes the ratio of $\alpha$-rare elements in the entire stream~$S$, i.e.,
\(
\rho_\alpha = \frac{\left |\{x\  \vert\  x \mbox{ is } \alpha\mbox{-rare in } S\}\right |}{\DEC(S)}\ .
\)
\end{definition}

Our algorithm follows the method used by~\cite{DM02} to estimate
$\alpha$-rarity in the unbounded model.
The estimation is based on the fact that the $\alpha$-rarity is equal to the portion of min-hash functions that their min-value appears exactly $\alpha$ times in the stream.

However, in order to estimate rarity over sliding windows,
one needs to estimate the ratio of min-hash functions
of which the min-value appears exactly $\alpha$ times {\em within the window}.
Our algorithm builds a smooth-histogram for $\DEC$ in order
to partition the stream into buckets,
such that each two consecutive buckets  have approximately the same
number of distinct elements.
In addition,
we sample the bucket using a min-wise hash, and count
the  $\alpha+1$ last occurrences of the sampled element $x_i$ in the bucket.
We estimate the  $\alpha$-rarity of the window by
calculating the fraction of min-hash functions
of which the appropriate min-value $x_i$ appears exactly $\alpha$ times {\em within the window}.
Due to feasibility reasons we use approximated min-wise hashes,
and prove that this estimation is an $\epsilon$-approximation
of the $\alpha$-rarity of the current window
(up to a pre-specified additive precision).
The {\em semi-smooth} algorithm \rarity\ for $\alpha$-rarity is defined in Fig.~\ref{alg:rarity}.
\begin{figure}[htb]
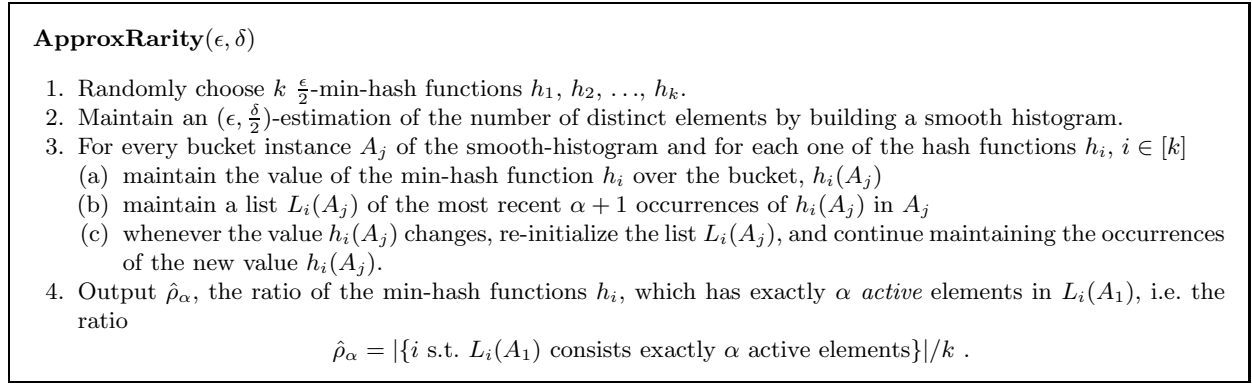

\begin{framed}
\rarity$(\epsilon, \delta)$
\small
\begin{enumerate}
\item Randomly choose $k$ $\frac\epsilon2$-min-hash functions $h_1$, $h_2$, $\ldots$, $h_k$.
\item Maintain an $(\epsilon,\frac\delta2)$-estimation
 of the number of distinct elements by building a smooth histogram.
\item For every bucket instance $A_j$ of the smooth-histogram and for each one of the hash functions
$h_i$, $i\in[k]$
  \begin{enumerate}
  \item maintain the  value of the min-hash function $h_i$ over the bucket, $h_i(A_j)$
  \item maintain a list $L_i(A_j)$ of the most recent $\alpha+1$ occurrences of $h_i(A_j)$ in $A_j$
  \item whenever the value $h_i(A_j)$ changes, re-initialize the list $L_i(A_j)$, and continue maintaining
  the occurrences of the new value $h_i(A_j)$.
  \end{enumerate}
\item Output $\hat \rho_\alpha$, the ratio of the min-hash functions $h_i$, which has exactly $\alpha$ {\em active} elements in  $L_i(A_1)$, i.e.\ the ratio
\[
\hat\rho_\alpha = |\{i \text{ s.t. } L_i(A_1) \text{ consists exactly }\alpha \text{ active elements}\}| /k\ .
\]
\end{enumerate}
\vspace{-10pt}
\end{framed}
\vspace{-10pt}
\caption{Semi-smooth algorithm for $\alpha$-rarity}\label{alg:rarity}
\vspace{-10pt}
\end{figure}

The \rarity\ algorithm provides an $(\epsilon,\delta)$-approximation for the $\alpha$-rarity problem, up to an additive error of $\epsilon$.
As proven by Datar et al.~\cite{DM02}, the ratio of min-hash functions that have exactly $\alpha$ active elements in the window is an estimation of $\rho_\alpha$. This is true even when using the
min-value of the inclusive bucket $A_1$ rather than the min-value of the current windows $W$.
\begin{theorem}\label{thm:rarity}
The semi-smooth algorithm (Fig.~\ref{alg:rarity}) is an
$(\epsilon,\delta)$-approximation for the $\alpha$-rarity problem, up to an additive precision.
\end{theorem}
\begin{proof}
For the sake of simplicity we treat the multisets $A_1$, $W$, etc., as sets.
Let $R_\alpha$ be the set of elements which are $\alpha$-rare in the window $W$.
Following Lemma~\ref{probAW}, with $R_\alpha \subseteq A_1$,
\[\Pr [h_i(A_1)=h_i(R_\alpha)] = \frac{|R_\alpha \cap A_1|}{|R_\alpha\cup A_1|}\pm\frac\epsilon2
=\frac{|R_\alpha|}{|A_1|}\pm\frac\epsilon2.
\]

The algorithm outputs an approximation of $\Pr\big[L_i(A_1)$
consists of exactly $\alpha$ active elements$\big]$, which
equals to $\Pr[h_i(A_1)=h_i(R_\alpha)]$, since $h_i(A_1)=h_i(R_\alpha)$ if and only if
$L_i(A_1)$ consists of $\alpha$~active elements.
Let~$x_i$ be the element which minimizes~$h_i$ on~$A_1$,  $h(x_i) = h(A_1)$.
If the number of active elements in $L_i(A_1)$ is not $\alpha$,
then $x_i \not\in R_\alpha$, thus $h(A_1) \ne h(R_\alpha)$.
For the other direction,
if $h_i(A_1) = h_i(R_\alpha)$
then $L_1(A_1)$ counts the number of occurrences of $x_i$ in the bucket,
and since~$x_i \in R_\alpha$, it appears exactly $\alpha$~times within the window.

We build a smooth-histogram for~$\DEC$ by using the algorithm of
Kane et al.~\cite{KNW10} as an approximation of~$\DEC$ for the unbounded model
(see Theorem~3 in~\cite{BO07}).
The smooth-histogram guarantees\footnote{Actually, it guarantees even a better bound,
specifically, $(1-\tfrac\epsilon2)|A_1| \le |W| \le |A_1|$.} that ${(1-\epsilon)|A_1| \le |W| \le |A_1|}$, thus
\[
\frac{|R_\alpha|}{|A_1|} \le \frac{|R_\alpha|}{|W|} = \rho_\alpha \qquad , \qquad
\frac{|R_\alpha|}{|A_1|} \ge 
(1-\epsilon)\frac{|R_\alpha|}{|W|} \ge (1-\epsilon)\rho_\alpha \ .
\]

Therefore, estimating the ratio~$\rho_\alpha$ using $k$~hash functions
results with a value $(1\pm\epsilon)\rho_\alpha\pm\frac\epsilon2$
up to some additive error $\epsilon'$ determined by~$k$.
Finally, using Chernoff's inequality we can bound the
additive error so that $\epsilon' < \frac\epsilon2$, except for probability~$\frac\delta2$.
In order to achieve the desired precision we require
$k = \Omega (\frac1{\epsilon^{2}}\log\frac1\delta)$,
and the estimation satisfies
\begin{eqnarray*}
\hat\rho_\alpha &\in& (1\pm\epsilon)\rho_\alpha\pm\epsilon\ ,
\end{eqnarray*}
except for probability at most $\delta$.\enlargethispage{2ex}
This concludes the correctness of the algorithm.
\qed
\end{proof}

\smallskip\noindent\textbf{Memory Usage.}
The memory consumption of the  \rarity\ algorithm is as follows.
Maintaining a smooth histogram for $\DEC$
is done using the method of Kane et al.~\cite{KNW10} as the underlying algorithm for DEC in the unbounded model, with memory
$\tilde O\big( (\log u+\frac{1}{\epsilon^2})\frac{1}{\epsilon}\log N\log\frac1\delta+\frac{1}{\epsilon}\log^2N\big)$;
$k$ seeds for the $\frac\epsilon2$-min-hash functions: $O(k \log\frac1\epsilon\log u)$;
 {Saving a list $L_i$ and a value $h_i$ for each bucket $A_j$ and for $i\in[k]$:}
$O([\log u+\alpha\log N]\frac{k}{\epsilon}\log N )$.

We note that this improves the {\em expected} memory bound of
Datar et al.~\cite{DM02} into a
{\em worst case} bound of the same magnitude (up to a $\log\log N$ term).
In most of the practical cases $\log u $ and $\log N $ are very close, and we can assume that
$\log u = O(\log N)$. In that case, the space complexity is
$\tilde O\left(\frac{k}{\epsilon}\alpha\log^2N\right)$ bits, with $k=\Omega(\frac1{\epsilon^{2}}\log\frac1\delta)$, and the
 time complexity is
$\tilde O\left(\frac{k\alpha}{\epsilon}\log N\right)$
calculations per element, suppressing $poly(\log\frac{1}{\epsilon}$, $\log\log N)$ terms.

\subsection{A semi-smooth estimation of streams similarity}\label{sec:similarity}
In this section we present an algorithm for calculating the {\em similarity}
of two streams $X$ and $Y$. As in the case of the rarity,
the similarity property is known not to be smooth, however
we are able to design a semi-smooth algorithm that estimates it.
We maintain a smooth-histogram of the distinct elements count
in order to partition each of the streams, and
sample each bucket of this partition
using a min-hash function. We compare the ratio of sample agreements in order to estimate the similarity of the two streams.

\begin{definition}
The (Jaccard) {\em similarity} of two streams, $X$ and $Y$ is given by
\(
S(X,Y)=\frac{|X\cap Y|}{|X \cup Y|}\ .
\)
\end{definition}
Recall that for two streams $X$ and $Y$, a reasonable estimation of $S(X,Y)$
is given by the number of min-hash values they agree on~\cite{DM02}. In other words,
let $h_1, h_2 , \dotsc, h_k$ be a family of $\epsilon$-min hash functions and let
$$\hat S(X,Y) = \left | \left \{ i \in [k] \mbox{ s.t. } h_i(X)=h_i(Y) \right \} \right | / k \; ,$$ then
$\hat S(X,Y) \in (1\pm\epsilon)S(X,Y) + \epsilon (1+p)$,  with success probability at least $1-\delta$,
where $p$ and $\delta$ are determined by $k$.
Based on this fact, Datar et al.~\cite{DM02} showed an algorithm for estimating similarity in the sliding window model,
that uses expected memory of $O(k(\log\frac{1}{\epsilon}+\log N))$ words with
$k=\Omega(\frac1{\epsilon^{3}p}\log{\frac{1}\delta})$.
%
Using smooth-histograms, our algorithm reduces the expected memory bound into a worst-case bound. The semi-smooth algorithm \similarity\ is rather straightforward and is given in Fig.~\ref{alg:similarity}.

\begin{figure}[htb]
\begin{framed}
{\bf \similarity$(\epsilon,\delta)$}
\small
\begin{enumerate}
\item Randomly choose $k$ $\epsilon'$-min-hash functions, $h_1, \ldots, h_k$. The constant $\epsilon'$ will be specified later, as a function of the desired precision $\epsilon$.
\item For each stream ($X$ and $Y$)
maintain an $(\epsilon',\frac\delta2)$-estimation
of the number of distinct elements by building a smooth histogram.
\item For each stream and for each bucket instance $A_1, A_2, \dotsc,$
separately calculate the values of each of the min-hash functions $h_i$, $i=1\ldots k$.
\item Let $A_X$ ($A_Y$) be the first smooth-histogram bucket that includes the current window
$W_X$ ($W_Y$) of the stream $X$ ($Y$). Output the ratio of hash-functions $h_i$ which agree on the minimal value, i.e.,
$$\hat\sigma(W_X,W_Y) =  \left |\left\{ i \in [k] \mbox{ s.t. } h_i(A_X)=h_i(A_Y)\right \} \right | / k\ .$$
\end{enumerate}
\vspace{-10pt}
\end{framed}
\vspace{-10pt}
\caption{A semi-smooth algorithm for estimating similarity}\label{alg:similarity}
\vspace{-10pt}
\end{figure}

\begin{theorem}\label{thm:similarity}
The semi-smooth algorithm for estimating similarity (Fig.~\ref{alg:similarity}), is an $(\epsilon,\delta)$-approximation for the similarity problem, up to an additive precision.
\end{theorem}

\begin{proof}
Following Lemma~\ref{probAW}, 
\[
Pr [h_i(A_X)=h_i(A_Y)]= \frac{|A_X\cap A_Y|}{|A_X \cup A_Y|}\pm \epsilon'\mbox{.}
\]

For convenience, once again we treat
 buckets $A_X,A_Y,W_X,W_Y$ as sets.
Notice that we can write $A_X = W_X \cup (A_X \setminus W_X)$ and that
$0 \le |A_X\setminus W_X|\le \frac{\epsilon'}{1-\epsilon'}|W_X|$, which follows from
the guarantee of the smooth-histogram that
 $(1-\epsilon')|A_X| \le |W_X| \le |A_X|$ (and same for $A_Y$ and $W_Y$).
Using elementary set operations, we can estimate $|W_X\cup W_Y|$ using $|A_X\cup A_Y|$,
\begin{eqnarray*}
|W_X\cup W_Y| \le &|A_X \cup A_Y| & \le |W_X\cup W_Y| + \frac{\epsilon'}{1-\epsilon'} |W_X| +\frac{\epsilon'}{1-\epsilon'} |W_Y| \\
&& \le |W_X\cup W_Y| + 2 \frac{\epsilon'}{1-\epsilon'}|W_X\cup W_Y| \\
&&=
\frac{1+\epsilon'}{1-\epsilon'}|W_X\cup W_Y| \ .
\end{eqnarray*}
In addition, any two sets $S,Q$ always satisfy $\frac{|S\cap Q|}{|S\cup Q|} = \frac{|S|+|Q|}{|S\cup Q|}-1$, thus
the similarity estimation satisfies
\begin{eqnarray*}
 \frac{|A_X\cap A_Y|}{|A_X \cup A_Y|} &=& \frac{|A_X| + |A_Y|}{|A_X \cup A_Y|}-1
  \le 
 \frac{\frac{1}{1-\epsilon'}|W_X| + \frac{1}{1-\epsilon'}|W_Y|}{|W_X \cup W_Y|}-1
 = \frac{1}{1-\epsilon'}\frac{|W_X\cap W_Y|}{|W_X\cup W_Y|} + \frac{\epsilon'}{1-\epsilon'},\quad \text{ and }\\[2ex]
\frac{|A_X\cap A_Y|}{|A_X \cup A_Y|} &\ge&
\frac{|W_X\cap W_Y|}{\frac{1+\epsilon'}{1-\epsilon'}{|W_X\cup W_Y|}}
 = \frac{1-\epsilon'}{1+\epsilon'}\frac{|W_X\cap W_Y|}{|W_X\cup W_Y|},
\end{eqnarray*}
Finally, setting $\epsilon' \le \epsilon/2$ gives an estimation 
$\hat \sigma(W_X,W_Y) \in (1\pm \epsilon)S(W_X,W_Y) \pm \epsilon$, 
up to an additional
additive error, which can be arbitrarily decreased using Chernoff's bound, by increasing $k$.
Specifically, this additional error is bounded by $O(\epsilon)$ when $k=\Omega(\frac1{\epsilon^{2}}\log\frac1\delta)$, with success probability at least ${1-O(\delta)}$.
\qed
\end{proof}

\smallskip\noindent\textbf{Memory Usage.}
Let us summarize the memory consumption of the \similarity\ algorithm.
{ Maintaining a smooth histogram for $\DEC$:} $\tilde O\big( (\log u+\frac{1}{\epsilon^2})\frac{1}{\epsilon}\log N \log\frac1\delta+\frac{1}{\epsilon}\log^2 N\big)$;
{$k$ seeds for $\epsilon/2$-min-hash functions:} $O(k \log\frac1\epsilon\log u)$;
{Keeping the hash value for each $h_i$:} $O\big(k\frac{1}{\epsilon}\log N\log u\big)$.

Our algorithm improves the currently known
{\em expected} bound~\cite{DM02} into a {\em worst case} bound of the same magnitude (up to a $\log\log N$ term).
Taking $k=\Omega(\frac1{\epsilon^{2}}\log\frac1\delta)$ and assuming $\log u =O (\log N)$, we achieve a memory bound of
$\tilde O\big(k\frac{1}{\epsilon}\log^2 N\big)$,
with $\tilde O(k\frac1\epsilon\log N)$ calculations per element,
suppressing $poly(\log \frac1\epsilon, \log\log N)$ elements.

\section{Conclusions}\label{sec:conclusions}
We have shown the first polylogarithmic algorithm for
identifying $L_2$ heavy-hitters up to $1\pm\epsilon$ precision, over sliding windows.
Our result supplies another insight about the relations between
the unbounded and sliding window models, for the central question of heavy-hitters.
As the $L_p$-heavy-hitters problem is more difficult for  larger $p$, and for $p>2$
there cannot exist a polylogarithmic solution, our algorithm provides a small-memory solution
for the ``strongest'' $L_p$~norm.

Although our main concern was the $L_2$ norm,
the algorithm can easily be extended for any $L_p$ with $0<p\le 2$.
Moreover, a polylogarithmic approximation of the top-$k$ problem in sliding window
is immediate using our methods.

The tools shown in this paper can be applied to many
other properties, if there exists a smooth function
which is correlated to the target function.
We have shown how to employ the same techniques in order to obtain
an efficient sliding window algorithm for the similarity and $\alpha$-rarity problems, 
with essentially the same memory consumption as the current state of the art,
however, our bound applies for the \emph{worst case} rather than holds only in expectation.
We believe that our method can be used to 
improve the memory efficiency 
of many other sliding-window algorithms
for non-smooth properties.

\section*{Acknowledgments}
V.B. is supported in part by DARPA grant N660001-1-2-4014.
R.O. is supported in part by NSF grants CNS-0830803; CCF-0916574; IIS-1065276; CCF-1016540; CNS-1118126; CNS-1136174; US-Israel BSF grant 2008411, OKAWA Foundation Research Award, IBM Faculty Research Award, Xerox Faculty Research Award, B. John Garrick Foundation Award, Teradata Research Award, and Lockheed-Martin Corporation Research Award. This material is also based upon work supported by the Defense Advanced Research Projects Agency through the U.S. Office of Naval Research under Contract N00014-11-1-0392. The views expressed are those of the author and do not reflect the official policy or position of the Department of Defense or the U.S. Government.




\appendix
\section*{Appendix}

\section{The \charikar\ Algorithm}\label{sec:varCharikar}
In this section we describe the \charikar\ algorithm 
and prove several of its properties.
Let us sketch the details of the {\sc CountSketch} algorithm as defined in~\cite{CCF02}.
{\sc CountSketch} is defined by three parameters $(t,b,k)$ such that the algorithm takes
space $O(tb+k)$, and if $t=O(\log {\frac{n}\delta})$ and $b\ge \max(8k, 256\frac{L_2}{\epsilon^2\phi_k^2})$
then the algorithm outputs any element with frequency at least $(1+\epsilon)\phi_k$,
except with probability $\delta$. 
$\phi_k$ is the frequency of the $k$th-heavy element, 
and $L_2$ is the $L_2$-frequency norm of the entire ($n$-element) stream. 
The algorithm works by computing, for each element $i$, an approximation $\hat n_i$ of its frequency. 
The scheme guarantees that with high probability, for every element $i$, $|\hat n_i - n_i| < 8\frac{L_2(S)}{\sqrt{b}}$ 
(see Lemma 4 in~\cite{CCF02}).

For
$0 < {\epsilon',\gamma,\delta} \le1$ 
define $(\gamma,\epsilon',\delta)$-{\sc CountSketch}$_b$ as
the algorithm {\sc CountSketch},
setting $k=\frac{1}{\gamma^2}+1$ and
letting $b = \frac{256}{\gamma^2\epsilon'^2}$ (the parameter $t$ remains as in the original scheme).
The choice of $k$ follows from the following known fact.
\begin{lemma}\label{lem:notallheavy}
There are at most $\frac1{\gamma^2}$ elements with frequency higher than $\gamma L_2$.
\end{lemma}
\begin{proof}
Assume that there are $m$ elements with frequency higher than $\gamma L_2$. It follows that
 $L_2 = (\sum_{j=1}^u n_j^2)^{1/2} \ge \sqrt{m} \cdot \gamma L_2$. Clearly, $m \le \frac1{\gamma^2}$.
\qed
\end{proof}
\noindent
Setting $k=\frac{1}{\gamma^2}+1$
ensures that the output list is large enough to contain all the elements with frequency
$\gamma L_2$ or more.

However,
\charikar\ does not guarantee anymore to output all the elements with frequency higher than
$(1+\epsilon')\fk$ and no element of frequency less than $(1-\epsilon')\fk$
(Lemma~5 of~\cite{CCF02}), since the value of $b$ might not satisfy the conditions of that lemma.

We can still follow the analysis of~\cite{CCF02} and claim that
the frequency approximation of each element is still bounded (Lemma~4 of~\cite{CCF02}),
\begin{lemma}\label{lem:varCharikarFreqApprox}
With probability at least~$1-\delta$, for all elements $i\in[u]$ in the stream~$S$,
\[
|\hat n_i - n_i| < 8\frac{L_2(S)}{\sqrt{b}} <\tfrac12\gamma\epsilon' L_2(S)
\]
where $\hat n_i$ is the approximated frequency of~$i$ calculated by \charikar, and~$n_i$
is the real frequency of the element~$i$.
\end{lemma}
The proof is immediate from~\cite{CCF02}. 
The above lemma allows us to bound the frequencies of the outputted elements

\begin{proposition}\label{prop:allHigh}
The $(\gamma,\epsilon',\delta)-$\charikar\ algorithm outputs all the elements whose frequency is at least
$(1+\epsilon')\gamma L_2(S)$.
\end{proposition}
\begin{proof}
An element is not in the output list only if there are (at least)
$k$ elements with higher approximated frequency.
Due to Lemma~\ref{lem:varCharikarFreqApprox}, any element $i$ with frequency
$n_i > (1+\epsilon')\gamma L_2(S)$ has an estimated frequency of at least $\hat n_i \ge (1+\frac12\epsilon')\gamma L_2(S)$, so it can be replaced only by an element with frequency higher than $\gamma L_2(S)$, however, there are at most $k$ elements with $n_i \ge \gamma L_2(S)$, specifically,
at most $k-1$ elements other than $i$ itself,
which completes the proof.
\qed
\end{proof}

The memory consumption of \charikar\ is bounded
by $O((k+ b)\log\frac{|S|}\delta))$~\cite{CCF02}, which in our case gives
$O(\frac1{\gamma^2\epsilon'^2}\log\frac{|S|}\delta)$.

\end{document}